\documentclass[runningheads,a4paper]{llncs}

\usepackage{algorithm}
\usepackage{algorithmic}
\usepackage{color,xcolor}
\usepackage{xspace}
\usepackage{subfig}
\usepackage{hyperref}
\usepackage{amsmath}
\usepackage{amssymb,graphicx}
\usepackage{bbm}

\begin{document}

\title{Modeling Intratumor Gene Copy Number Heterogeneity using  Fluorescence in Situ Hybridization data}

\titlerunning{Modeling Intratumor Gene Copy Number Heterogeneity using FISH data}  
%
\author{Charalampos E. Tsourakakis\inst{1}}
\authorrunning{Charalampos E. Tsourakakis}   
\institute{ 
Carnegie Mellon University, USA \\
\email{ctsourak@math.cmu.edu} 
}


\maketitle

\newcommand{\Prob}[1]{{{\bf{Pr}}\left[{#1}\right]}}
\newcommand{\reminder}[1]{{\textsf{\textcolor{red}{[#1]}}}}
\newcommand{\QED}{ \hfill {$\Box$}}
\newcommand{\vectornorm}[1]{\left|\left|#1\right|\right|}
\newcommand{\field}[1]{\mathbb{#1}} 
\renewcommand{\vec}[1]{{\mbox{\boldmath$#1$}}}
\newcommand{\Abs}[1]{{\left|{#1}\right|}}
\newcommand{\Lone}[1]{{\left\|{#1}\right\|_{L^1}}}
\newcommand{\Linf}[1]{{\left\|{#1}\right\|_\infty}}
\newcommand{\Norm}[1]{{\left\|{#1}\right\|}}
\newcommand{\Mean}[1]{{\mathbb E}\left[{#1}\right]}
\newcommand{\Var}[1]{{\mathbb Var}\left[{#1}\right]}
\newcommand{\Floor}[1]{{\left\lfloor{#1}\right\rfloor}}
\newcommand{\Ceil}[1]{{\left\lceil{#1}\right\rceil}}

\begin{abstract}\footnote{Topic: Cancer Genomics. Keywords: intra-tumor heterogeneity, evolutionary dynamics, cancer phylogenetics, Markov chains, simulation, FISH}
Tumorigenesis is an evolutionary process which involves a significant number
of genomic rearrangements typically coupled with changes in the gene copy number profiles
of numerous cells.  
Fluorescence {\it in situ} hybridization (FISH) is a  cytogenetic technique which 
allows counting copy numbers of genes in single cells. The study of cancer
progression using FISH data has received considerably less attention compared to other 
types of cancer datasets. 

In this work we focus on inferring likely tumor progression pathways using publicly
available FISH data. 
We model the evolutionary process as a Markov chain in the positive integer
cone $\field{Z}_+^g$ where $g$ is the number of genes examined with FISH. 
Compared to existing work which oversimplifies reality by assuming
independence of copy number changes \cite{pennington,pennington2}, our model is able to capture dependencies.
We model the probability distribution of a dataset with hierarchical log-linear models, 
a popular probabilistic model of count data. Our choice provides an attractive trade-off 
between parsimony and good data fit.
We prove a theorem of independent interest which provides necessary and sufficient conditions for reconstructing 
oncogenetic trees \cite{desper}. Using this theorem we are able to capitalize on the 
wealth of inter-tumor phylogenetic methods. We show how to produce tumor phylogenetic trees 
which capture the dynamics of cancer progression. 
We validate our proposed method on a breast tumor dataset. 
 
\end{abstract}

\section{Introduction}
\label{sec:intro}
Tumors are heterogeneous masses which exhibit cellular and genomic differences \cite{heselmeyer,marusyk,navin2,navin3}.
Cell-by-cell assay measurements allow us to study the phenomenon of tumor heterogeneity. 
Fluoresence {\it in situ} hybridization (FISH) is a cytogenetic 
technique which allows us to study {\em gene copy number heterogeneity} within a single tumor. 
It is used to count the copy number of DNA probes for specific genes or chromosomal regions. 
Understanding how tumor heterogeneity  progresses is a major problem
with significant potential impact on therapeutics. 
An example which illustrates the importance of heterogeneity in
therapeutic resistance is found in chronic myelogenous leukemia (CML).
The presence of a subpopulation of leukemic cells of a given type 
significantly influences the response to therapies based on imatinib mesylate,
causing the eventual relapse of the disease \cite{marusyk}.

\begin{figure*}[htp]
    \centering 
    \includegraphics[width=0.4\textwidth]{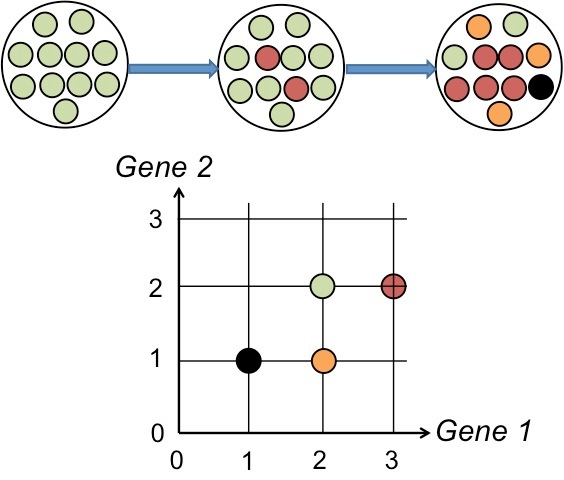}
    \caption{Copy number genetic diversity within a tumor cell population after three hypothetical stages of tumorigenesis.
	The four observed states of cells are shown and coded with colors in the positive integer cone $\field{Z}_+^2$.
	 } 
     \label{fig:fig1}
\end{figure*}

In this work we study the phenomenon of {\em gene copy number heterogeneity} within a single tumor. 
An illustration of the phenomenon of gene copy number heterogeneity is shown in Figure~\ref{fig:fig1}
which shows a hypothetical cell population of eleven somatic human cells at three stages of tumorigenesis,
whose succession is indicated by arrows. 
We shall describe the state of a cell as a two dimensional point $(g_1,g_2)$ in the positive integer cone $\field{Z}_+^2$,
with the obvious meaning that the cell has $g_1,g_2$ copies of gene 1 and gene 2 respectively.  
Initially, all cells are in the healthy diploid state with respect to their gene copy numbers.
At the second stage, nine cells are in state (2,2) and two cells in state (3,2). 
At the last stage, only 2 cells are in the healthy state (2,2). 
Five, three and one cell are in states (3,2), (2,1) and (1,1) respectively. 
The four observed states are shown and coded with colors in $\field{Z}_+^2$.
Figure~\ref{fig:fig1} illustrates what we observe in many tumors: existence of multiple progression states within a single tumor.

Despite the large amount of research work on modeling tumor progression using different types of tumor datasets, e.g., 
\cite{hainke}, the study of FISH datasets has received considerably less attention. 
Specifically, two early studies of FISH datasets were limited to either two \cite{pennington,pennington2} or 
three probes \cite{martins2012evolutionary}.
Pennington et al. \cite{pennington,pennington2} develop novel computational methods for analyzing FISH 
data. Specifically, they consider a random walk on the positive 
integer cone $\field{Z}_+^2$ where at each step a coordinate $i \in \{1,2\}$ is picked uniformly 
at random and is modified by $\Delta x \in \{0,1,-1\}$
with probabilities $1-p_{i,1}-p_{i,-1}, p_{i,1}, p_{i,-1}$, $i=1,2$ respectively.
Given this model they optimize a likelihood-based objective over all possible trees
and parameters  $\{p_{i,1}, p_{i,-1}\}$. 
Recently, Chowdhury et al. \cite{ismb13} proposed   a general procedure which can treat any number of probes. 
They reduce the study of the progression of FISH probe cell count patterns to the rectilinear minimum spanning tree problem. 

\textit{Paper contributions and roadmap.} In this paper we achieve the following contributions:

\begin{itemize}

 \item We introduce a novel approach to analyzing FISH datasets. The main features of our approach 
 are its probabilistic nature which provides an attractive trade-off between parsimony and 
 expressiveness of biological complexity and the reduction of the problem to the well-studied inter-tumor phylogeny inference problem. 
 The former allows us to capture complex dependencies between factors while the latter 
 opens the door to a wealth of available and established theoretical methods which exist for the inter-tumor phylogeny inference problem.

 \item  We prove Theorem~\ref{thrm:thrm1}  which provides necessary and sufficient conditions for the unique reconstruction of an oncogenetic tree \cite{desper}. 
 Based on the theorem's conditions, we are able to capitalize  on the wealth of inter-tumor phylogenetic methods. 
 However, the result is of independent interest and introduces a set of interesting combinatorial questions. 
 
 \item We validate our proposed method on a publicly available breast cancer dataset.  
 
\end{itemize}

The outline of this paper is as follows: Section~\ref{sec:methods} presents our proposed
methods. Section~\ref{sec:experiments} performs an experimental evaluation of our methods on 
a breast cancer FISH dataset and an extensive biological analysis of the findings. 
Finally, Section~\ref{sec:concl} concludes the paper by a discussion and a brief summary.

\section{Proposed Method}
\label{sec:methods}
In Section~\ref{subsec:probabilisticmodel} we model the probability distribution 
of FISH data with hierarchical log-linear models and show how to learn 
the parameters of the model for a given FISH dataset. 
In Section~\ref{subsec:oncotrees} we prove  Theorem~\ref{thrm:thrm1} 
which provides necessary and sufficient conditions to uniquely reconstruct 
an oncogenetic tree \cite{desper}. 
We capitalize on the theorem to harness the wealth of available 
methods for inter-tumor phylogenetic inference 
methods \cite{beer4,beer5,beer6,desper,desper2,hainke,beer3,heydebreck}.
Finally in Section~\ref{subsec:simulation} we present our proposed method. 

We will make the same simplifying assumptions with existing work \cite{pennington,ismb13}, namely that only single gene duplication and loss events take place
and that the cell population is fixed. 
In what follows, let $\mathcal{D}=\{x_1,\ldots,x_n\}, x_i \in \field{Z}_+^g$ be the input FISH dataset which measures the copy numbers of $g$ genes 
in $n$ cells taken from the same tumor.

\subsection{Model and Fitting}
\label{subsec:probabilisticmodel} 

\paragraph{ Probabilistic Model:} Let $X_j$ be an integer-valued random variable which expresses
the copies of the $j$-th gene with domain $I_j$, $j=1,\ldots,g$. 
We model the joint probability distribution of the random vector $(X_1,\ldots,X_g)$ as 

\begin{equation} 
\label{eq:modeleq}
\Prob{x} = \frac{1}{Z} \prod_{A \subseteq [g]} e^{\phi_A(x)}
\end{equation}

\noindent where $x=(x_1,\ldots,x_g) \in I=I_1 \times I_2 \times ... \times I_g$ 
is a point of the integer positive cone $\field{Z}_+^g$ and 
$Z$ is a normalizing constant, also known as the partition function, which ensures that the distribution is a proper probability distribution, i.e.,
$ Z = \sum_{x \in I} \prod_{A \subseteq [g]} e^{\phi_A(x)}$.
Each potential function $\phi_A$ depends only on the variables in the subset $A$ and is parameterized by a set of weights $w_A$. 
To illustrate this, assume $g=2$ and $I=\{0,1\} \times \{0,1\}$.
Then, 

\begin{align*}
\log \Prob{x} &= w_0 + w_{(1)0} \mathbbm{1}\{x_1=0 \} + w_{(1)1} \mathbbm{1}\{x_1=1\} +w_{(2)0} \mathbbm{1}\{x_2=0\} \\
	      &+ w_{(2)1} \mathbbm{1}\{x_2=1\}+  w_{(12)00} \mathbbm{1}\{x_1=0,x_2=0\} + w_{(12)01} \mathbbm{1}\{x_1=0,x_2=1\}  \\ 
	      &+ w_{(12)10} \mathbbm{1}\{x_1=1,x_2=0\} + w_{(12)11} \mathbbm{1}\{x_1=1,x_2=1\} - \log Z,
\end{align*}

\noindent where $w_{Ax}$ $A \subseteq \{1,2\}, x \in \{0,1\}^{|A|}$ are the parameters of the model. 
This probability distribution captures the effects of different factors through 
parameters $w_{Ax}, A \in \{ \{1\}, \{2\} \}, x \in \{0,1\}$ and pairwise interactions through parameters 
$w_{(12)x}, x \in \{00,01,10,11\}$. 
In general, two variables $X_i,X_j$ are defined to be directly associated 
if there exists at least one non-zero (or bounded away significantly from zero) 
parameter including the two variables. 
We define $X_i,X_j$ to be indirectly associated if there exists a chain of overlapping direct 
associations that relate $X_i,X_j$. 

\noindent We impose the following restriction on the probabilistic model shown in equation~\eqref{eq:modeleq}:
If $A \subseteq B$ and $w_A=0$ then $w_B=0$.  
This restriction reduces significantly the size of the parameter space, 
but allows to express complex dependencies not captured by existing work \cite{pennington,pennington2}. 
Since a typical FISH dataset contains detailed measurements  for a handful of genes from few hundred cells
the combination of these two features is crucial to avoid overfitting and 
obtain biological insights at the same time. 
Furthermore, it is worth emphasizing that in terms of biological interpretation the assumption is natural:
if a set $A$ of genes does not interact, then any superset of $A$ maintains that property. 
This class of models are known as hierararchical log-linear models \cite{bishop}.

\paragraph{ Learning the Parameters:} 

Learning the parameters $w$ of a hierarchical log-linear model is a well-studied problem, e.g., \cite{bishop}. 
An extensive survey of learning methods can be found in \cite{thesischmidt}. 
Schmidt et al. \cite{schmidt2} propose to maximize a penalized log-likelihood of the dataset $\mathcal{D}$ 
where the penalty is an overlapping group $l_1$-regularization term. 
Specifically, a spectral projected gradient method is proposed as a sub-routine for solving the following regularization problem:

\begin{equation} 
\label{eq:schmidt} 
\min_w -\sum_{i=1}^n \log \Prob{x_i|w} + \sum_{A \subseteq S} \lambda_A ||w_A||_1.
\end{equation} 

\subsection{Unique Reconstruction of Oncogenetic Trees}
\label{subsec:oncotrees}

Our main theoretical result in this section is motivated by the following natural sequence of questions: 

{\em Can we use any of the existing inter-tumor progression methods \cite{hainke} on the intra-tumor progression problem? 
How will the resulting tree capture the evolutionary dynamics of cancer progression, i.e., how 
do we enforce that state  $(\ldots,i,\ldots)$ is reached either through $(\ldots,i+1,\ldots)$ or $(\ldots,i-1,\ldots)$ 
given our single gene duplication and loss event assumption?}

An answer to this question is given in Section~\ref{subsec:simulation}.
Motivated by our intention to capitalize on inter-tumor phylogenetic methods such as \cite{desper,desper2},
we consider a fundamental problem concerning oncogenetic trees \cite{desper}. What are the necessary and sufficient
conditions to reconstruct them?  
Theorem~\ref{thrm:thrm1} is likely to be of independent interest and contributes 
to the understanding of oncogenetic trees \cite{desper}. We briefly review the necessary definitions to state our result. 
Let $T=(V,E,r)$ be an oncogenetic tree, i.e., a rooted branching\footnote{Each vertex has in-degree at most one and 
there are no cycles.}, on $V$ and let $r \in V$ be the root of $T$.
Given a finite family $\mathcal{F}=\{A_1,...A_q\}$ of sets of vertices, i.e., $A_i \subseteq V(T)$ for $i=1,\ldots,q$,
where each $A_i$ is the vertex set of a rooted sub-branching of $T$, what are the necessary and sufficient conditions, if any,  
which allow us to uniquely reconstruct $T$?

\begin{theorem}
\label{thrm:thrm1}
The necessary and sufficient conditions to uniquely reconstruct the branching $T$ from the family $\mathcal{F}$ are
the following:
\begin{enumerate}
		\item For any two distinct vertices $x,y \in V(T)$ such that $(x,y) \in E(T)$, 
                there exists a set $A_i \in \mathcal{F}$ such that $x \in A_i$ and $y \notin A_i$.
	        \item For any two distinct vertices $x,y \in V(T)$ such that $y \nprec x$ and $x \nprec y$\footnote{We use the notation $u \prec  v$ ($u \nprec v$)
                to denote that $u$ is (not) a descendant of $v$ in $T$.}
	        there exist sets $A_i,A_j \in \mathcal{F}$ such that $x \in A_i$, $y \notin A_i$ and $x \notin A_j$ and $y \in A_j$.
\end{enumerate} 
\end{theorem}

\begin{proof} 
First we prove the necessity of conditions 1,2 and then their sufficiency to reconstruct $T$. 
In the following we shall call a branching $T$ {\it consistent} with the family set $\mathcal{F}$ if 
all sets $A_i \in \mathcal{F}$ are vertices of rooted sub-branchings of $T$. 

\begin{figure*}
\centering
\label{tab:tab1} 
  \begin{tabular}{c|c}  
   \includegraphics[width=0.4\textwidth]{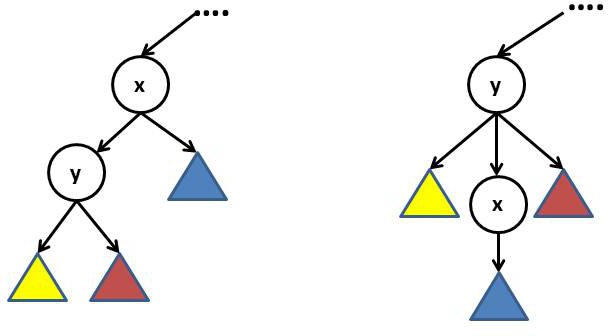} & \includegraphics[width=0.4\textwidth]{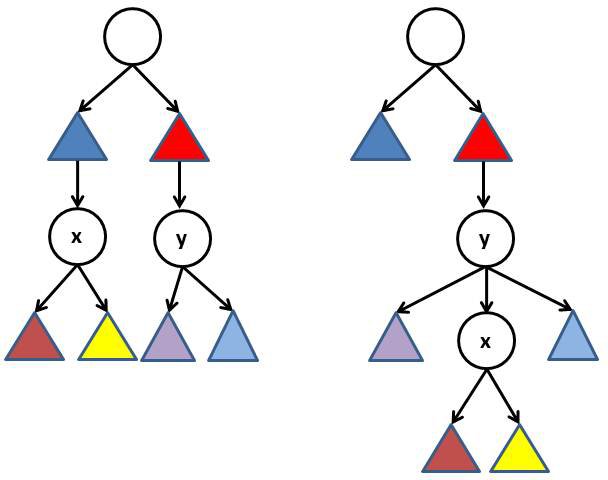}  \\  
    (a)   &  (b)  \\  
  \end{tabular}
  \caption{Illustration of necessity conditions of Theorem~\ref{thrm:thrm1}, (a) condition 1 (b) condition 2.}
\end{figure*}

\underline{Necessity:} For the sake of contradiction, assume that Condition 1 does not hold. Then, the two branchings 
shown in Figure 2(a) are both consistent with $\mathcal{F}$ and therefore we cannot reconstruct $T$. 
Similarly, assume that Condition 2 does not hold. Specifically assume that 
for all $j$ such that $x \in A_j$, then $y \in A_j$ too (for the symmetric case
the same argument holds). Then, both branchings in Figure 2(b) are consistent with $\mathcal{F}$ 
and therefore $T$ is not reconstructable. 

\underline{Sufficiency:} Let $x \in V(T)$  and $P_x$ be the path from the root to $x$, i.e., $P_x=\{r,\ldots,x\}$. 
Also, define $F_x$ to be the intersection of all sets in the family $\mathcal{F}$ that contain vertex $x$,
i.e., $F_x =\underset{\mbox{ $x \in A_i \in \mathcal{F}$ }}{ \bigcap A_i }$. 
We prove that $F_x = P_x$. Since by definition, $P_x \subseteq F_x$ we need to show that $F_x \subseteq P_x$.
Assume that the latter does not hold. 
Then, there exists a vertex $v \in V(T)$ such that $v \notin P_x, v \in F_x$. 
We consider the following three cases. 
 
\noindent 
\underline{$\bullet$ {\sc Case 1} ($x \prec v$):} Since every $A_i \in \mathcal{F}$ is the vertex set of a rooted sub-branching,  $v \in P_x$
by definition.

\noindent 
\underline{$\bullet$ {\sc Case 2} ($v \prec x$):} By condition 1 and an easy inductive argument, 
there exists $A_i$ such that $x \in A_i, v \notin A_i$. Therefore, $v \notin F_x$. 

\noindent 
\underline{$\bullet$ {\sc Case 3} ($x \nprec v, v \nprec x$):}
By condition 2, there exists $A_i$ such that  $x \in A_i$ and $v \notin A_i$.
Therefore, in combination with the definition of $F_x$ we obtain $v \notin F_x$. 

In all three cases, we obtain a contradiction and therefore $v \in F_x \Rightarrow v \in P_x$,
showing that $F_x = P_x$. 
Given this fact, it is easy to reconstruct the branching $T$. We sketch the algorithm: compute for each $x$ the set 
$F_x$ and from $F_x$ reconstruct the ordered version of the path $P_x$, i.e.,
$(r \rightarrow v_1 \rightarrow .. \rightarrow x)$ using sets in $\mathcal{F}$ whose 
existence is guaranteed by condition 1.  $\qed$
\end{proof}

\begin{figure*}
\centering
\label{tab:tab1} 
  \begin{tabular}{cc}  
   \includegraphics[width=0.45\textwidth]{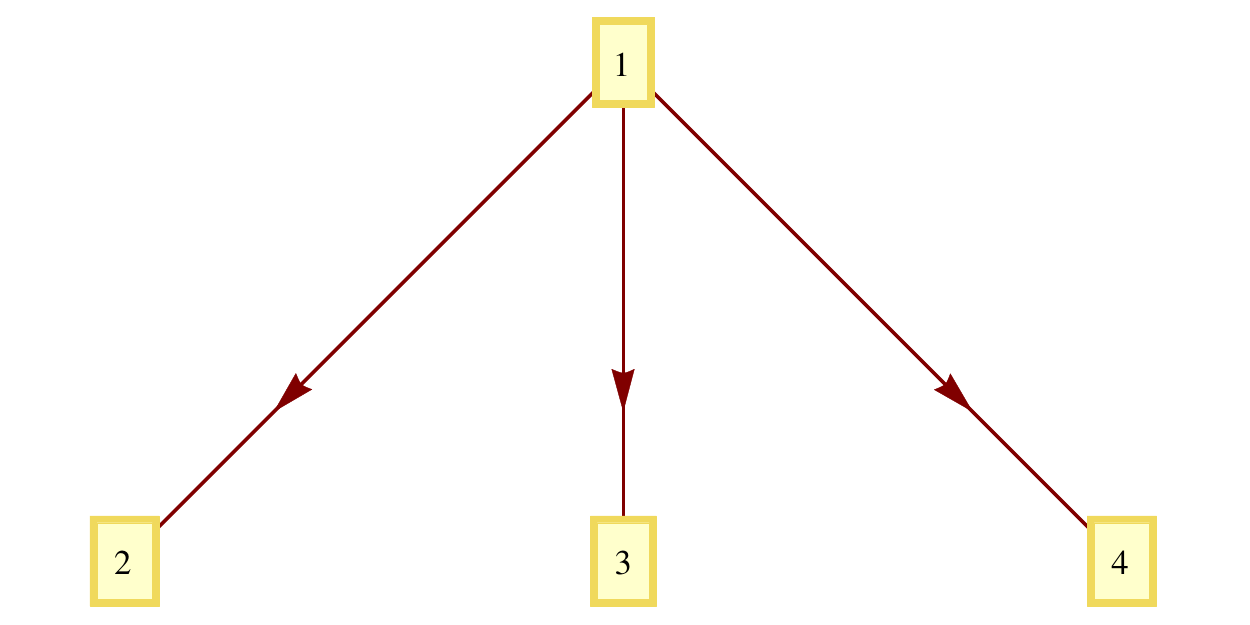} & \includegraphics[width=0.6\textwidth]{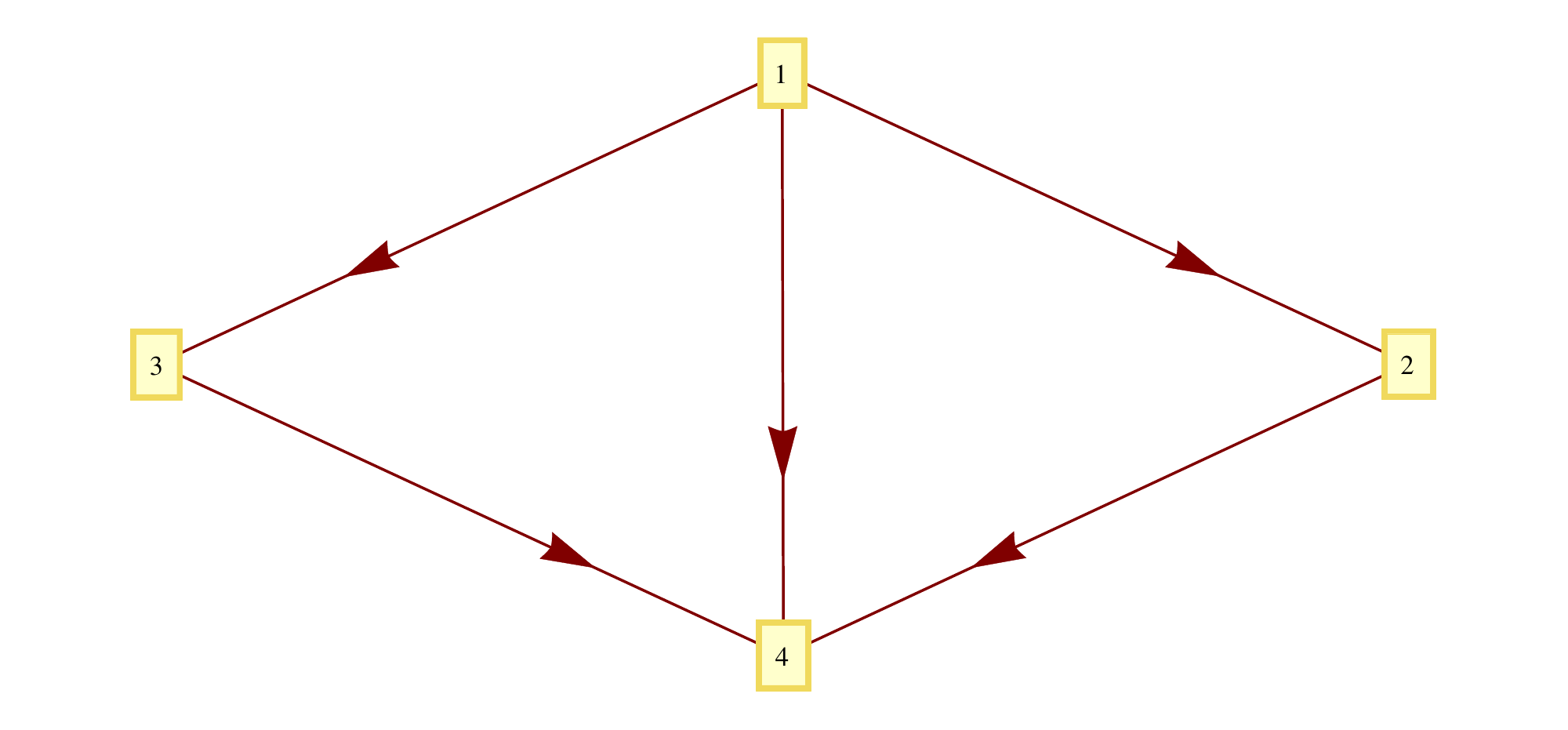}  \\  
    (a)   &  (b)  \\  
  \end{tabular}
  \caption{\label{fig:counterex} Extending Theorem~\ref{thrm:thrm1} to oncogenetic directed acyclic graphs  (DAGs) \cite{hainke} is not possible. 
  For instance, the oncogenetic tree and the DAG are indistinguishable in terms of the generated families of sets.}
\end{figure*}

A natural question is whether one can extend Theorem~\ref{thrm:thrm1} to more complex classes
of oncogenetic models, such as directed acyclic graphs (DAGs) \cite{hainke}. The answer is negative. 
For instance, the oncogenetic tree with edges $1 \rightarrow 2,1 \rightarrow 3$ 
and the DAG with edges $1 \rightarrow 2, 1 \rightarrow 3, 2 \rightarrow 3$ 
are indistinguishable. Another example is shown in Figure~\ref{fig:counterex}. 

Oncogenetic tree reconstruction algorithms \cite{desper}  and more generally established inter-tumor cancer progression methods \cite{hainke} 
receive as input a family of sets, where each set represents the set of mutations observed in a {\it single} tumor. 
Recall, that in the case of intra-tumor cancer progression the typical 
input is a multiset of points in the positive integer cone where each point is the copy-gene number state of a given cell. 
For instance, for the tumor shown in the final step of tumorigenesis in Figure~\ref{fig:fig1},
the input would be $\mathcal{D} = \{ (2,2) \times 2,(3,2) \times 5,(2,1) \times 3,(1,1) \times 1 \}$. 

Based on the insights from the proof of Theorem~\ref{thrm:thrm1}, we 
convert a FISH dataset to a dataset suitable for inter-tumor cancer progression inference.
Specifically, we assume we are given a FISH dataset and an algorithm $f$ which infers an 
evolutionary model of cancer progression from several tumors. Notice that $f$ could be 
any inter-tumor phylogenetic method, see \cite{hainke}. 

For each cell in state $x=(x_1,\ldots,x_g)$  we generate a family of sets of ``mutations'' as follows:
for every gene $i \in [g]$ whose number of copies $x_i$ is greater than  2 we generate 
$x_i-1$ sets in order to enforce that the gene has $c+1$ copies if at a previous stage
had $c$ copies, $c=2,..,x_i-1$. 
For instance if gene $i$ has 4 copies, we generate the sets $\{gene-i-mut-2\}, \{gene-i-mut-2,gene-i-mut-3\},
\{gene-i-mut-2,gene-i-mut-3,gene-i-mut-4\}$ which show how the gene
obtained 2 extra copies 
The case $x_i<2$  is treated in a similar way. 
Finally, we generate a set $s_f$ for each cell which contains all mutations 
that led to state $x$. For instance, for the state $(0,3)$ 
$s_f =\{ gene-1-mut-2, gene-1-mut-1, gene-1-mut-0, gene-2-mut-2, gene-2-mut-3 \}$.
Upon creating the dataset $\mathcal{F}$ we use it as input to $f()$,
an existing intra-tumor phylogenetic method, see \cite{hainke}.

Using the conversion above, based on  condition 1 of Theorem~\ref{thrm:thrm1} the output of 
an inter-tumor phylogenetic method will capture the dynamic nature of the process, 
which will be consistent with our assumptions of single gene duplication and loss events 
and Ockham's razor, e.g.,, the evolutionary sequence $2 \rightarrow 3 \rightarrow 4$
rather than $2 \rightarrow 3 \rightarrow 4 \rightarrow 5 \rightarrow 4$.

\subsection{Progression Inference}
\label{subsec:simulation} 

Define $\mathcal{B} = [\min_{i \in [n]} x_{i1}, \max_{i \in [n]} x_{i1}] \times .. \times [\min_{i \in [n]} x_{ig}, \max_{i \in [n]} x_{ig}]$ to
be the minimum enclosing box of $\mathcal{D}$, where $x_{ij}$ is the number of copies of gene $j$ in the $i$-th cell, $i \in [n], j \in [g]$.  
Given the observed data we can calculate the empirical probability $\tilde{\pi}(s)$ of any state $s \in \mathcal{B}$
as the fraction $\frac{|\{q: q \in \mathcal{D}, q=s|\}|}{n}$. 
The number of states in $\mathcal{B}$ grows exponentially fast for any typical FISH dataset. 
We summarize parsimoniously this distribution as described in Section~\ref{subsec:probabilisticmodel}. 
Specifically, we learn the parameters $w$ of the hierararchical log-linear model by 
maximizing the overlapping $l_1$ penalized log-likelihood of equation~\eqref{eq:schmidt}
as described in \cite{thesischmidt}.
We allow only second-order interactions between factors.
It is worth mentioning that $k$-way interactions, $k \geq 3$ 
can be embedded in the model as well, see Chapter 6 \cite{thesischmidt},
but we prefer not to avoid overfitting. Alternatively we can allow higher order
interactions but then a penalty term for the model complexity (e.g., AIC, BIC)
should be taken into account.
Let $\pi$ be the distribution specified by the learned parameters. 
We define a Metropolis-Hastings chain with stationary distribution $\pi$ \cite{peresbook} . 
Initially, all $n$ cells will be in the diploid state $(2,\ldots,2)$.
Notice that all we need to compute during the execution of the chain are ratios
of the form $\pi(x)/\pi(y)$, which saves us from the computational cost of computing 
the normalization constant $Z$. 
We  simulate the chain $k$ times in order to draw $m\geq 1$ samples from the probability distribution.
Finally we use the conversion described in Section~\ref{subsec:oncotrees} to infer a tumor phylogeny. 

There exists a subtle issue that arises in practice: there exist states of $\mathcal{B}$ which 
are not observed in the dataset $\mathcal{D}$. We surpass this problem by adding one fictitious 
sample to each state $b \in \mathcal{B}$. From a Bayesian point of view this is equivalent to smoothing the data with an 
appropriately chosen Dirichlet prior. 

To summarize, our proposed method consists of the following steps: (1)
Given a FISH dataset $\mathcal{D}$ we learn the parameters of a hierararchical log-linear model with  pairwise potentials. 
(2) Given the learned parameters we can compute the probability distribution on $Z^g$. 
 Let $\pi$ be the resulting distribution. We define a Metropolis-Hastings chain  with stationary distribution $\pi$. 
 Initially  cells are in the healthy diploid state  $ (2,\ldots,2)$. 
(3) Draw $m \geq 1$ samples from the probability distribution by running the Metropolis-Hastings chain simulation $m$ times.
(4) Convert the resulting FISH samples to inter-tumor phylogenetic datasets by following the procedure of Section~\ref{subsec:oncotrees}. 
(5) Use an inter-tumor phylogenetic method  \cite{hainke} to infer a tumor phylogenetic tree.

Finally, an interesting perspective on our modeling which makes a conceptual connection to \cite{ismb13} is the following: 
upon learning the parameters of the hierararchical log-linear model, 
the probability distribution over $\mathcal{B}$ assigns implicitly weights on the edges of the positive integer di-grid $\field{Z}^g$
(each undirected edge of the grid is substituted by two directed edges) according to the Metropolis-Hastings chain.
Therefore, both our method and \cite{ismb13} assign weights to the edges of the positive integer grid.
This perspective opens two natural research directions which we leave open for future research. 
First, instead of simulating the Markov chain, one proceed could find an appropriate subgraph of the weighted di-grid, e.g., 
a maximum weighted branching rooted at the diploid state. 
Secondly, it is natural to ask whether there exists a natural probabilistic interpretation of the method in \cite{ismb13}.

\section{Experimental Results}
\label{sec:experiments}
 
\textit{Experimental Setup:} In this paper, we show the results of validating our method on a breast cancer dataset
from a collection of publicly available FISH datasets which can be found at \url{ftp://ftp.ncbi.nlm.nih.gov/pub/FISHtrees/data/}. 

We used the following third-party publicly available code in our implementations: 
hierarchical log-linear fitting code\footnote{Schmidt's code does not scale well to more than 6-7 variables.} 
\cite{thesischmidt}, distance based oncogenetic trees \cite{desper2}, 
FISH progression trees \cite{ismb13} and {\it graphviz} for visualization purposes. 
Our routines are implemented in MATLAB. The number of simulations was set to $m=10$. 
We experimented both with smaller values for parameter $m$ ($m \geq 2$) 
and the choice of log-linear fitting method, see \cite{delaportas},
and we found that our results are robust. 

Table 2 provides a short description of the six genes that are analyzed in our trees. 
The breast tumor dataset consists of 187 points in $\field{Z}_+^6$.
 \begin{table*}
 \centering
 \begin{tabular}{|c|c|c|}  \hline 
  Gene         &  Cytogenetic Band & Description          \\ \hline   
  {\it cox-2}  &  1q25.2-q25.3  & oncogene              \\ \hline 
  {\it myc}    &  8q24        & oncogene               \\ \hline 
  {\it ccnd1}  &  11q13       & oncogene                \\ \hline 
  {\it cdh1}   &  16q22.1     & tumor suppressor  gene     \\ \hline 
  {\it p53}    &  17p13.1     & tumor suppressor  gene      \\ \hline 
  {\it znf217} &  20q13.2      & oncogene                   \\ \hline 
 \end{tabular}
 \label{tab:genedescription}
  \caption{  Genes are shown in the first column and their cytogenetic positions in  the second. 
   The third column describes whether a gene is an oncogene or a tumor suppressor gene. }
\end{table*}


\begin{figure*}
  \includegraphics[width=0.95\textwidth]{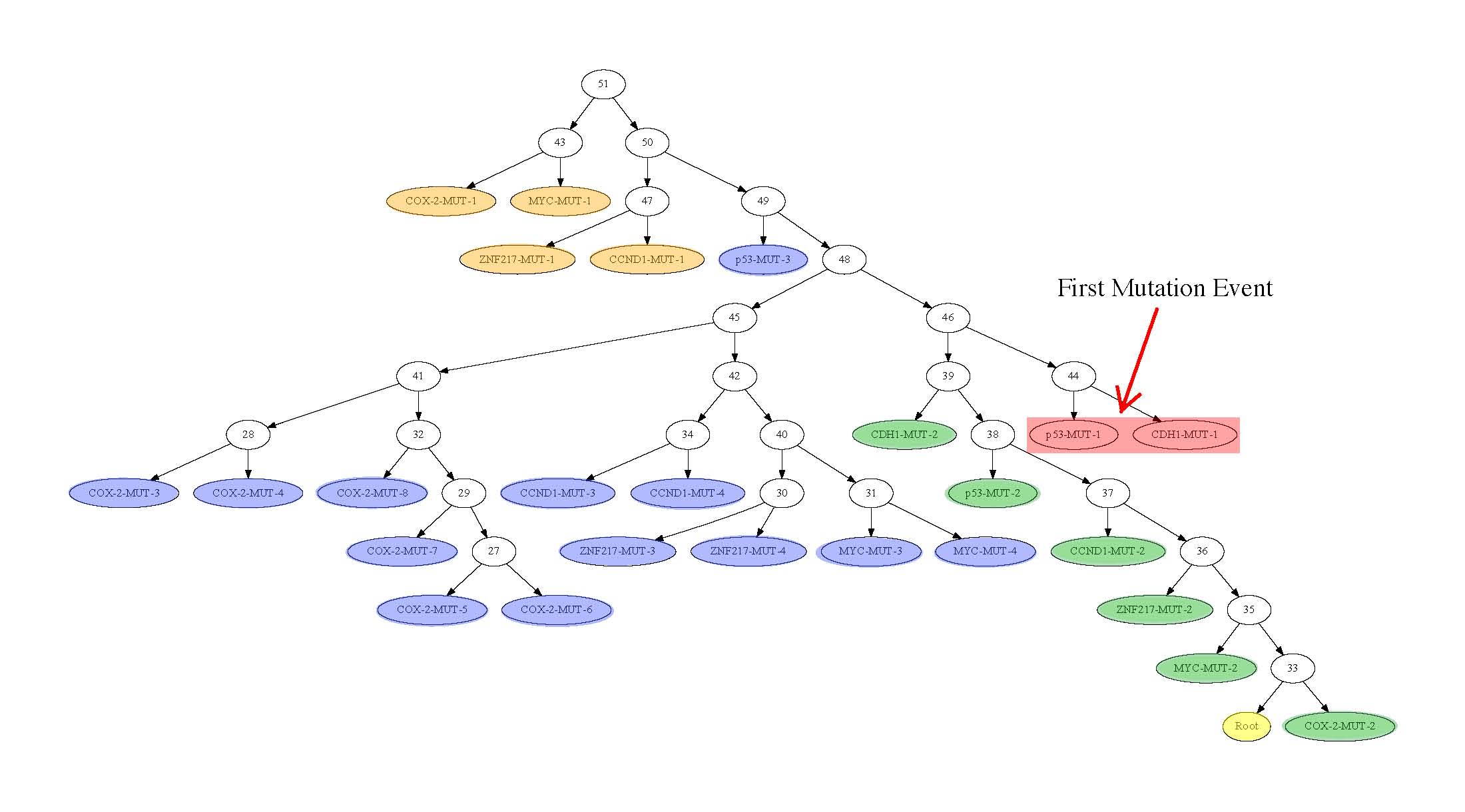}
  \caption{\label{fig:trees2} (best viewed on screen) Cancer phylogenetic tree for a breast cancer tumor obtained by our method. Leaves of the tree are colored in the following way:
  the root-event leaf is in yellow; the first change in the copy number profile is in red color; 
  euploid states are in green; 
  states with copy number gain/loss are shown in blue/orange.  
  The first changes are losses of one gene copy 
  of genes {\it p53} and {\it cdh1}.}
\end{figure*}

\textit{Results and Analysis:} Figures~\ref{fig:trees2} and~\ref{fig:trees3} show the cancer phylogenetic trees of a ductal carcinoma {\it in situ} (DCIS)
obtained by our method and \cite{ismb13} respectively. 
Our tree is a distance based phylogenetic tree produced by using our reduction 
of intra-tumor phylogenetic inference to inter-tumor phylogenetic inference
as described in Section~\ref{subsec:simulation}. 
We observe that the tree of \cite{ismb13} does not explain the 27 different
states that appear in the dataset. For instance
the state $(4, 8, 4, 4, 2, 2, 4)$ accounts for 0.0053\% of the appearing 
states and is discarded by \cite{ismb13} but is taken into account by our method. 
Since there is no ground truth available to us it is hard to reach any indisputable
conclusion. However, we found that our findings are strongly supported 
by oncogenetic literature. 

\begin{figure*}
  \centering 
  \includegraphics[width=0.80\textwidth]{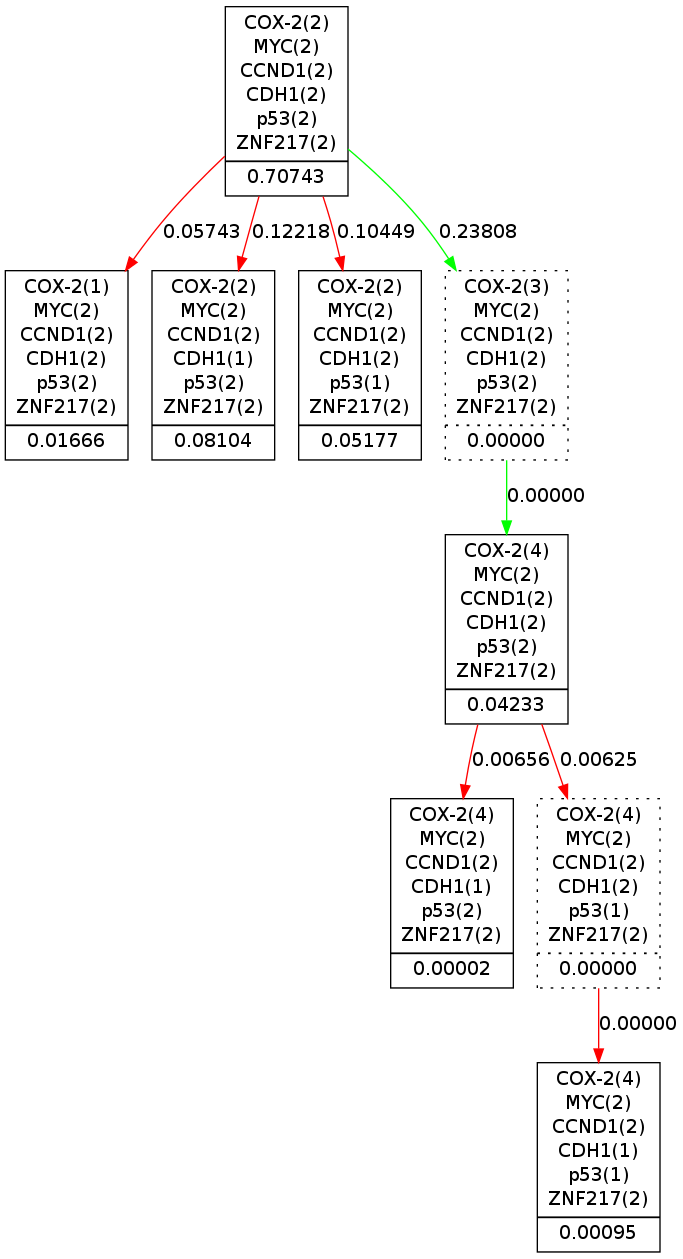}
  \caption{\label{fig:trees3} Cancer phylogenetic tree for a breast cancer tumor obtained by \cite{ismb13}. Nodes with dotted borders represent Steiner 
  nodes, i.e., states that do not appear in the dataset. Green and red edges model gene gain and loss respectively. The weight value on each edge 
  does not have the semantics of probability, but it is the rectilinear distance between the two connected states. See \cite{ismb13} for further details. 
  The weight on each node describes the fraction of cells in the FISH dataset with the particular copy number profile.}
\end{figure*}

The mutational events captured in our phylogenetic tree highlight putative sequential events 
during progression from ductal carcinoma in situ (DCIS) to invasive breast carcinoma. 
The first mutational events are highlighted in red. 
Initially one allele of {\it p53} and {\it cdh1} are lost.
Concurrent loss of {\it cdh1} function  and {\it p53} inactivation act 
synergistically in the formation, progression and metastasis of breast cancer \cite{Derksen06}. 
Moreover, following the first mutational events, the next changes occur in {\it ccnd1}, {\it myc} and {\it znf217}
which are oncogenes participating in cell cycle regulation, proliferation and cancer progression.
Specifically, as shown by single  invasive ductal carcinoma (IDC) cell analysis  copy number loss of {\it cdh1} is common in DCIS.
Furthermore, copy number gains of {\it myc} are a common feature in the transition from DCIS to IDC \cite{heselmeyer}. 
This is consistent with our results. 
The synergy between  {\it p53} and {\it cdh1} appears also in the 
tree of Figure~\ref{fig:trees3} but at the last stages of the progression. 
Finally, Figure~\ref{fig:associations} shows the inferred direct associations among genes.

\section{Conclusion}
\label{sec:concl}
In this work we  develop a novel approach to studying FISH datasets,
a type of dataset which has received considerably less attention 
to other types of cancer datasets. 
Compared to prior work we take a probabilistic approach which 
provides good data fit, avoids overfitting and captures complex
dependencies among factors. 
Motivated by our intention to capitalize on inter-tumor phylogenetic
methods we prove a theorem which provides necessary and sufficient conditions
for reconstructing oncogenetic trees \cite{desper}. 
Using these conditions, we show one way to perform intra-tumor phylogenetic inference 
by opening the door to the wealth of established inter-tumor phylogenetic techniques \cite{hainke}. 
We model the evolutionary dynamics as a Markov chain in the positive integer 
cone $\field{Z}_+^g$ where $g$ is the number of genes examined with FISH. 
Finally, we validate our approach to a breast cancer FISH dataset.

\begin{figure*}
 \centering
  \includegraphics[width=0.35\textwidth]{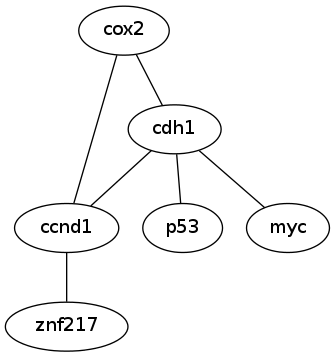}
  \caption{Direct associations among genes inferred from fitting a hierarchical log-linear model to a ductal carcinoma {\it in situ}  FISH dataset. }
  \label{fig:associations}
\end{figure*}

Our work leaves numerous problems open for future research.  
Improved models need to be developed that remove the simplifying assumption of a fixed cell population
and take the clonal evolution model into account \cite{nowell}, namely
cancer is initiated once multiple mutations occur in a random single cell 
which gives birth to the uncontrolled proliferation of cancerous cells. 
Secondly, clustering patients and finding consensus FISH progression trees per cluster 
is another interesting problem. 
Furthermore, we plan to experiment with (a) other choices of inter-tumor phylogenetic methods and (b)
fitting approaches that allow higher order interactions but will also account for the increased 
complexity of the resulting model. 
Finally, using features from our inferred trees as features for classification is  an interesting question, see \cite{ismb13}.


\section{Acknowledgments} 

The author would like to thank Prof. Russell  Schwartz and 
Maria Tsiarli for helpful discussions on intratumor heterogeneity
and the anonymous reviewers for their constructive feedback.

\end{document}